\documentclass[9pt,technote]{IEEEtran}

\usepackage{amsfonts}
\usepackage{cite}
\usepackage{graphicx}
\usepackage{psfrag}
\usepackage{subfigure}
\usepackage{stfloats}
\usepackage{amsmath}   
\usepackage{array}
\newtheorem{Tm}{Theorem}
\newtheorem{remark}{Remark}
\newtheorem{exmp}{Example}

\newcommand{\tr}[1]{\mathop{{\rm \bf tr}\left[#1\right]}\nolimits}

\newcounter{tempo_eq}

\begin{document}
\title{Likelihood Gradient Evaluation Using Square-Root Covariance Filters}
%
%
\author{M.V.~Kulikova
\thanks{Manuscript received September 14, 2007; revised May 9, 2008.}
\thanks{M.V.~Kulikova is with the School of Computational and Applied
Mathematics, University of the Witwatersrand, South Africa; Email:
Maria.Kulikova@wits.ac.za.} }

\markboth{PREPRINT}{}

\maketitle

\begin{abstract}
Using the array form of numerically stable square-root
implementation methods for Kalman filtering formulas, we construct a
new square-root algorithm for the log-likelihood gradient (score)
evaluation. This avoids the use of the conventional Kalman filter
with its inherent numerical instabilities and improves the
robustness of computations against roundoff errors. The new
algorithm is developed in terms of covariance quantities  and based
on the "condensed form" of the array square-root filter.
\end{abstract}


\IEEEpeerreviewmaketitle

\section{Introduction}
\PARstart{C}{onsider}  the discrete-time linear stochastic system
  \begin{eqnarray}
    x_{k} &=& F_k x_{k-1}+G_k w_k,   \label{eq2.1} \\
    z_k &=& H_k x_k+v_k, \quad k=1,\ldots, N, \label{eq2.2}
  \end{eqnarray}
where $x_k \in \mathbb R^n$ and $z_k \in \mathbb R^m$ are,
respectively, the state and the measurement vectors; $k$ is a
discrete time, i.e. $x_k$ means $x(t_k)$. The noises  $w_k \in
\mathbb R^q$, $v_k \in \mathbb R^m$ and the initial state $x_0 \sim
{\cal N}(\bar x_0,\Pi_0)$ are taken from mutually independent
Gaussian distributions with zero mean and covariance matrices $Q_k$
and $R_k$, respectively, i.e. $w_k \sim {\cal N}(0,Q_k)$, $v_k \sim
{\cal N}(0,R_k)$. Additionally, system~(\ref{eq2.1}), (\ref{eq2.2})
is parameterized by a vector of unknown system parameters $\theta
\in \mathbb R^p$, which needs to be estimated. This means that the
entries of the matrices $F_k$, $G_k$, $H_k$, $Q_k$, $R_k$ and
$\Pi_0$ are functions of $\theta \in \mathbb R^p$. However, for the
sake of simplicity we will suppress the corresponding notations
below, i.e instead of $F_k(\theta)$, $G_k(\theta)$, $H_k(\theta)$,
$Q_k(\theta)$, $R_k(\theta)$ and $\Pi_0(\theta)$ we will write
$F_k$, $G_k$, $H_k$, $Q_k$, $R_k$ and $\Pi_0$.

Solving the parameter estimation problem by the method of maximum
likelihood requires the maximization of the likelihood function (LF)
with respect to unknown system parameters. It is often done by using
a gradient approach where the computation of the likelihood gradient
(LG) is necessary. For the state-space system~(\ref{eq2.1}),
(\ref{eq2.2}) the negative Log LF is given as~\cite{Schweppe1965}:
\begin{displaymath}
L_{\theta}\left(Z_1^N\right)= \frac{1}{2} \sum \limits_{k=1}^N
\left\{
 \frac{m}{2}\ln(2\pi)
+\ln\left(\det R_{e,k}\right) + e_k^T R_{e,k}^{-1}e_k \right\}
\end{displaymath}
where $Z_1^N=[z_1,\ldots, z_N]$ is $N$-step measurement history and
$e_k$ are the innovations,  generated by the discrete-time Kalman
filter (KF), with zero mean and covariance matrix $R_{e,k}$. They
are $e_k = z_k-H_k \hat x_{k|k-1}$ and
$R_{e,k}=H_kP_{k|k-1}H_k^T+R_k$, respectively. The KF
 defines the one-step ahead predicted state estimate $\hat x_{k|k-1}$
and the one-step predicted error covariance matrix $P_{k|k-1}$.

Straight forward differentiation of the KF equations is a direct
approach to the Log LG evaluation, known as a "score". This leads to
a set of $p$ vector equations, known as the {\it filter sensitivity
equations}, for computing $\partial{\hat
x_{k|k-1}}/\partial{\theta}$, and a set of $p$ matrix equations,
known as the {\it Riccati-type sensitivity equations}, for computing
$\partial{P_{k|k-1}}/\partial{\theta}$.

Consequently, the main disadvantage of the standard approach is the
problem of numerical instability of the conventional KF, i.e
divergence due to the lack of reliability of the numerical
algorithm. Solution of the matrix Riccati equation is a major cause
of numerical difficulties in the conventional KF implementation,
from the standpoint of computational load as well as from the
standpoint of computational errors~\cite{Grewal2001}.

The alternative approach can be found in, so-called, square-root
filtering algorithms. It is well known that numerical solution of
the Riccati equation tends to be more robust against roundoff errors
if Cholesky factors or modified Cholesky factors (such as the
$U^TDU$-algorithms~\cite{Bierman77}) of the covariance matrix are
used as the dependent variables. The resulting KF implementation
methods are called square-root filters (SRF). They are now generally
preferred for practical use~\cite{Grewal2001},
\cite{KaminskiBryson1971}, \cite{KailathSayed2000}. For more
insights about numerical properties of different KF implementation
methods we refer to the celebrated paper of Verhaegen and Van
Dooren~\cite{VerhaegenDooren1986}.

Increasingly, the preferred form for algorithms in many fields is
now the array form~\cite{Kailath_slides}. Several useful SRF
algorithms for KF formulas formulated in the array form have been
recently proposed by Park and Kailath~\cite{ParkKailath1995}. For
this implementations the reliability of the filter estimates is
expected to be better because of the use of numerically stable
orthogonal transformations for each recursion step. Apart from
numerical advantages, array SRF algorithms appear to be better
suited to parallel and to very large scale integration (VLSI)
implementations~\cite{ParkKailath1995}, \cite{Lee1993}.

The development of numerically stable implementation methods for KF
formulas has led to the hope that the Log LG (with respect to
unknown system parameters) might be computed more accurately. For
this problem, a number of questions arise:
\begin{itemize}
\item  Is it possible to
extend reliable array SRF algorithms to the case of the Log LG
evaluation?
\item If such methods exist, will they inherit the advantages from
the source filtering implementations? In particular, will they
improve the robustness of the computations against roundoff errors
compared to the conventional KF technique? The question about
suitability for parallel implementation is beyond the scope of this
paper.
\end{itemize}

The first attempt to answer these questions belongs to Bierman {\it
et al.}~\cite{BBVP1990}. The authors used the square-root
information filter, developed by Dyer and McReynolds~\cite{Dyer1969}
and later extended by Bierman~\cite{Bierman77}, as a source filter
implementation and constructed the method for score evaluation. The
algorithm was developed in the form of measurement and time updates.
However, the accuracy of the proposed method has not been
investigated.

In contrast to the main result of~\cite{BBVP1990}, we focus on the
dual class of KF implementation methods (that is the class of
covariance-type methods) and discuss the efficient Log LG evaluation
in square-root covariance filters. More precisely, we consider the
array form of the square-root covariance filter eSRCF introduced
in~\cite{ParkKailath1995}. The purpose of this paper is to design
the method for the Log LG evaluation in terms of the square-root
covariance variables, i.e. in terms of the quantities that appear
naturally in the eSRCF. This avoids the use of the conventional KF
with its inherent numerical instabilities and gives us an
opportunity to improve the robustness of the Log LG computation
against roundoff errors.

\begin{figure*}[!t]
\normalsize
\setcounter{tempo_eq}{\value{equation}}
\setcounter{equation}{7}
\begin{IEEEeqnarray}{l}
 Q_k \left[
\begin{array}{cc|c||cc|c}
R_k^{1/2} & 0 & -R_k^{-T/2}z_k \; \;&
\partial_{\theta_i}R_k^{1/2} & 0  &
\partial_{\theta_i}\left(-R_k^{-T/2}z_k\right) \\
P_{k}^{1/2}H_k^T & P_{k}^{1/2}F_k^T \;\; & P_{k}^{-T/2} \hat x_{k} &
\partial_{\theta_i}\left(P_{k}^{1/2}H_k^T\right) &
\partial_{\theta_i}\left(P_{k}^{1/2}F_k^T \right) &
\partial_{\theta_i}\left(P_{k}^{-T/2} \hat x_{k}\right)\\
0 & Q^{1/2}_kG_k^T  & 0 & 0
&\partial_{\theta_i}\left(Q^{1/2}_kG_k^T\right)  & 0
\end{array}
\right]  \nonumber \\
= {\;}\left[
\begin{array}{cc|c||cc|c}
R_{e,k}^{1/2} & \bar K_{p,k}^T  & - \bar e_k & X_i & Y_i & M_i \\
 0 & P_{k+1}^{1/2}  &
P_{k+1}^{-T/2} \hat x_{k+1} & N_i & V_i &  W_i  \\
\phantom{G^T_k} 0\phantom{G^T_k} &
\phantom{Q^{1/2}_k}0\phantom{G^T_k}  & \gamma_k & B_i & K_i & T_i
\end{array}
\right]. \label{gr-esrcf:1}
\end{IEEEeqnarray}
\setcounter{equation}{\value{tempo_eq}} \hrulefill \vspace*{4pt}
\end{figure*}

\section{Extended Square-Root Covariance Filter \label{eSRCF}}

To achieve our goal, we are first going to present the extended
square-root covariance filter (eSRCF), proposed
in~\cite{ParkKailath1995}, and second, we will derive the expression
for the Log LG evaluation in terms of the variables that are
generated by the eSRCF implementation.

{\it Notations to be used:} For the sake of simplicity, we denote
the one-step predicted state estimate as $\hat x_k$ and the one-step
predicted error covariance matrix as $P_k$. We use Cholesky
decomposition of the form $P_k=P_k^{T/2} P_k^{1/2}$, where
$P_k^{1/2}$ is an upper triangular matrix. Similarly, we define
$R_k^{1/2}$, $Q_k^{1/2}$, $R_{e,k}^{1/2}$. For convenience we will
write $A^{-1/2}=(A^{1/2})^{-1}$, $A^{-T/2}=(A^{-1/2})^T$ and
$\partial_{\theta_i}{A}$ implies the partial derivative of the
matrix $A$ with respect to the $i$th component of $\theta$, i.e
$\partial{A}/\partial{\theta_i}$.

In this paper, we deal with the "condensed form"\footnote{The
"condensed form" of filtering algorithms refers to the case when
implementation method for the KF formulas is not divided into the
measurement and time updates.} of the eSRCF~\cite{ParkKailath1995}:
Assume that $R_k
>0$. Given $\Pi_0^{1/2}$ and $\Pi_0^{-T/2}\bar x_0$, recursively update $P_{k}^{1/2}$ and
$P_{k}^{-T/2} \hat x_{k}$ as follows:
\begin{IEEEeqnarray}{r}
  Q_k \left[
\begin{array}{cc|c}
R_k^{1/2} & 0 & \ -R_k^{-T/2}z_k \\
P_k^{1/2}H_k^T & P_k^{1/2}F_k^T &
P_k^{-T/2} \hat x_k \\
0 & Q^{1/2}_kG_k^T &  0
\end{array}
\right] \nonumber \\
= \left[ \begin{array}{cc|c}
R_{e,k}^{1/2} & \bar K_{p,k}^T  & -\bar e_k \\
0 & P_{k+1}^{1/2} &  P_{k+1}^{-T/2} \hat x_{k+1} \\
0 & 0 &  \gamma_k
\end{array}
\right] \label{eq-esrcf}
\end{IEEEeqnarray}
where $Q_k$ is any orthogonal rotation that upper-triangularizes the
first two (block) columns of the matrix on the left-hand side
of~(\ref{eq-esrcf}); $\bar K_{p,k}=F_kP_{k} H_k^TR_{e,k}^{-1/2}$ and
$\bar e_k=R_{e,k}^{-T/2} e_k$.

One can easily obtain the expression for the negative Log LF in
terms of the eSRCF variables:
\begin{equation}
\label{eq:llfsqrt} L_{\theta}\left(Z_1^N\right)= \frac{1}{2} \sum
\limits_{k=1}^N \left\{
 \frac{m}{2}\ln(2\pi)
+ 2\ln\left(\det R_{e,k}^{1/2}\right) + \bar e_k^T \bar e_k
\right\}.
\end{equation}

Let $\theta=[\theta_1, \ldots, \theta_p]$ denote the vector of
parameters with respect to which the likelihood function is to be
differentiated. Then from~(\ref{eq:llfsqrt}), we have
\begin{equation}
\partial_{\theta_i} L_{\theta}\left(Z_1^N \right)= \sum
\limits_{k=1}^N  \left\{
  \partial_{\theta_i} \left[ \ln\left(\det R_{e,k}^{1/2}\right)\right] +
\frac{1}{2} \partial_{\theta_i} \left[ \bar e_k^T \bar e_k \right]
\right\}. \label{eq:llgsqrt}
\end{equation}

Taking into account that the matrix $R_{e,k}^{1/2}$ is upper
triangular, we derive
\begin{IEEEeqnarray}{rcl}
\partial_{\theta_i} \left[ \ln\left(\det R_{e,k}^{1/2}\right)\right] &
= &
\partial_{\theta_i} \left[
\sum \limits_{j=1}^{m} \ln \left( r_{e,k}^{jj} \right) \right]
 =  \sum \limits_{j=1}^{m} \left[ \frac{1}{r_{e,k}^{jj}} \cdot
\partial_{\theta_i} r_{e,k}^{jj} \right]  \nonumber \\
& = & \tr { R_{e,k}^{-1/2} \cdot
\partial_{\theta_i} R_{e,k}^{1/2}}, \hspace{2em} i=1, \ldots, p  \label{grad_det}
\end{IEEEeqnarray}
where the $r_{e,k}^{jj}$, $j=1, \ldots, m$ denote the diagonal
elements of the matrix $R_{e,k}^{1/2}$.

Substitution of~(\ref{grad_det}) into~(\ref{eq:llgsqrt}) yields the
result that we are looking for
\begin{equation}
  \label{grad-esrcf}
\partial_{\theta_i} L_{\theta}\left(Z_1^N
\right)=\sum \limits_{k=1}^N  \left\{
  \tr { R_{e,k}^{-1/2} \cdot
\partial_{\theta_i} R_{e,k}^{1/2}} \  + \
  \bar e_k^T \cdot \partial_{\theta_i}\bar e_k  \right\}.
\end{equation}

Ultimately, our problem is to compute Log LG~(\ref{grad-esrcf}) by
using the eSRCF equation~(\ref{eq-esrcf}). Before we come to the
main result of this paper, there are a few points to be considered.
As can be seen from~(\ref{grad-esrcf}), the elements $\bar e_k$ and
$R_{e,k}^{1/2}$ involved in the Log LG evaluation are obtained from
the underlying filtering algorithm directly, i.e.
from~(\ref{eq-esrcf}). No additional computations are needed. Hence,
our aim is to explain how the last two terms in the Log LG
expression, $\partial_{\theta_i}{\bar e_k}$ and
$\partial_{\theta_i}{R_{e,k}^{1/2}}$, can be computed using
quantities available from eSRCF~(\ref{eq-esrcf}).

\section{Suggested Square-Root Method for Score Evaluation \label{score-complex}}

We can now prove the following result.
\begin{Tm}
\label{theorem:1} Let the entries of the matrices $F_k$, $G_k$,
$H_k$, $Q_k$, $R_k$, $\Pi_0$ describing the linear discrete-time
stochastic system~(\ref{eq2.1}), (\ref{eq2.2}) be differentiable
functions of a parameter $\theta \in \mathbb R^p$. Then in order to
compute the Log LF and its gradient (with respect to unknown system
parameter $\theta$) the eSRCF, which is used to filter the data,
needs to be extended as follows. Assume that $R_k
>0$. Given the initial values $\Pi_0^{1/2}$, $\Pi_0^{-T/2}\bar x_0$
and $\partial_{\theta_i}{\Pi_0^{1/2}}$,
$\partial_{\theta_i}{\left(\Pi_0^{-T/2}\bar x_0\right)}$,
recursively update $P_{k}^{1/2}$, $P_{k}^{-T/2} \hat x_{k}$ and
$\partial_{\theta_i} P_{k}^{1/2}$, $\partial_{\theta_i} \left(
P_{k}^{-T/2} \hat x_{k}\right)$ as follows:
\begin{IEEEenumerate}
\item[{\bf I.}] Replace the eSRCF equation~(\ref{eq-esrcf})
by~(\ref{gr-esrcf:1}) where $Q_k$ is any orthogonal rotation that
upper-triangularizes the first two (block) columns of the matrix on
the left-hand side of~(\ref{gr-esrcf:1}).
\item[{\bf II.}] Having computed the elements of the
right-hand side matrix in~(\ref{gr-esrcf:1}), calculate  for each
$\theta_i$: \addtocounter{equation}{1}
\begin{equation}
\label{gr-esrcf:3} \left[
\begin{IEEEeqnarraybox}[][c]{c/c}
\partial_{\theta_i}{R_{e,k}^{1/2}} & \partial_{\theta_i}{\bar K_{p,k}^T}\\
0 &
\partial_{\theta_i}{P_{k+1}^{1/2}}
\end{IEEEeqnarraybox}
\right]= \Bigl[ \bar L_i^T+D_i+\bar U_i \Bigl] \left[
\begin{IEEEeqnarraybox}[][c]{c/c/c}
R_{e,k}^{1/2} & \bar K_{p,k}^T   \\
0 & P_{k+1}^{1/2}
\end{IEEEeqnarraybox}
\right],
\end{equation}
\setlength{\arraycolsep}{0.14em}
\begin{IEEEeqnarray}{r}
 \left[
\begin{IEEEeqnarraybox}[][c]{c}
-\partial_{\theta_i}{\bar e_k} \\
\partial_{\theta_i}{\left(P_{k+1}^{-T/2}\hat x_{k+1} \right) }
\end{IEEEeqnarraybox} \right]   =
\Bigl[ \bar L^T_i -\bar L_i \Bigl] \left[
\begin{IEEEeqnarraybox}[][c]{c}
-\bar e_k\\
 P_{k+1}^{-T/2}\hat x_{k+1}
\end{IEEEeqnarraybox}
\right] \nonumber \\
 +  \left[
\begin{IEEEeqnarraybox}[][c]{cc}
R_{e,k}^{1/2} \ & \  \bar K_{p,k}^T  \\
0 & P_{k+1}^{1/2}
\end{IEEEeqnarraybox}
\right]^{-T} \left[
\begin{IEEEeqnarraybox}[][c]{c}
B_i \\ K_i
\end{IEEEeqnarraybox}
\right] \gamma_k + \left[
\begin{IEEEeqnarraybox}[][c]{c}
M_i\\
W_i
\end{IEEEeqnarraybox}
\right] \label{gr-esrcf:4}
\end{IEEEeqnarray}
\setlength{\arraycolsep}{5pt} where $\bar L_i$, $D_i$ and $\bar U_i$
are strictly lower triangular, diagonal and strictly upper
triangular parts of the following matrix product:
\begin{equation}
\label{gr-esrcf:2} \!\!\! \left[
\begin{IEEEeqnarraybox}[][c]{c/c}
X_i & Y_i  \\
N_i & V_i
\end{IEEEeqnarraybox}
\right] \left[
\begin{IEEEeqnarraybox}[][c]{c/c}
R_{e,k}^{-1/2} & -R_{e,k}^{-1/2} \bar K_{p,k}^T P_{k+1}^{-1/2} \\
0 & P_{k+1}^{-1/2}
\end{IEEEeqnarraybox}
\right]= \bar L_i+D_i+\bar U_i.
\end{equation}
\item[{\bf III.}] Having determined $\bar e_k$,
$R_{e,k}^{1/2}$ and $\partial_{\theta_i} \bar e_k $,
$\partial_{\theta_i} R_{e,k}^{1/2}$ compute Log
LF~(\ref{eq:llfsqrt}) and Log LG~(\ref{grad-esrcf}).
\end{IEEEenumerate}
\end{Tm}

\begin{proof}
As discussed earlier, the main difficulty in score
evaluation~(\ref{grad-esrcf}) is to define $\partial_{\theta_i}
R_{e,k}^{1/2}$ and $\partial_{\theta_i} \bar e_k $  from the
underlying filter, i.e. from~(\ref{eq-esrcf}). We divide the proof
into two parts, first proving~(\ref{gr-esrcf:3}) for the
$\partial_{\theta_i} R_{e,k}^{1/2}$ evaluation and then
validating~(\ref{gr-esrcf:4}) for $\partial_{\theta_i} \bar e_k $.

{\bf Part I.} Our goal is to express  $\partial_{\theta_i}
R_{e,k}^{1/2}$ in terms of the variables that appear naturally in
the eSRCF implementation. First, we can note that the eSRCF
transformation in~(\ref{eq-esrcf}) has a form
\begin{displaymath}
QA=B
\end{displaymath} where $A$ is a rectangular matrix, and $Q$ is
an orthogonal transformation that block upper-triangularizes $B$. If
matrix $A$ is square and invertible, then given the matrix of
derivatives $A'_{\theta}=\displaystyle\frac{da_{ij}}{d\theta}$ we
can compute $B'_{\theta}$ as follows~\cite{BBVP1990}:
\begin{equation}
\label{bierman_sxema}
 B'_{\theta}=\left[ L^T+D+U \right] \: B
\end{equation}
where  $L$, $D$ and $U$ are, respectively, strictly lower
triangular, diagonal and strictly upper  triangular parts of the
matrix $QA'_{\theta}B^{-1}$.

However, this idea cannot be applied to the eSRCF because the matrix
to be triangularized, i.e. the first two (block) columns of the
matrix on the left-hand side of~(\ref{eq-esrcf}), is not square and,
hence, not invertible. By using the pseudoinversion (Moore-Penrose
inversion) we avoid this obstacle and generalize the scheme of
computations~(\ref{bierman_sxema}) to the case of
eSRCF~(\ref{eq-esrcf}).

To begin constructing the method for score evaluation, we augment
the matrix to be triangularized by $q$ columns of zeros. Hence, we
obtain
 \begin{equation}
  \label{grad:1}
\! Q_k \left[
\begin{array}{cc|c}
R_k^{1/2} & 0 & 0  \\
P_k^{1/2}H_k^T & P_k^{1/2}F_k^T &  0  \\
0 & Q^{1/2}_kG_k^T &  0
\end{array} \right] = \left[
\begin{array}{cc|c}
R_{e,k}^{1/2} & \bar K_{p,k}^T  &  0  \\
0 & P_{k+1}^{1/2}  &   0  \\
0 & 0 &  0 \end{array} \right].
\end{equation}

The matrices in~(\ref{grad:1}) have dimensions
$(m+n+q)\times(m+n+q)$. For the sake of simplicity, we denote the
left-hand side and the right-hand side matrices of~(\ref{grad:1}) as
$A_k$ and $B_k$, respectively. Then, by
differentiating~(\ref{grad:1}) with respect to the components of
$\theta$, we obtain
\begin{equation}
 \label{grad:2}
\partial_{\theta_i}Q_k \cdot A_k + Q_k \cdot \partial_{\theta_i}A_k =
\partial_{\theta_i}B_k.
\end{equation}

Multiplication both sides of~(\ref{grad:2}) by the pseudoinverse
matrix $B_k^+$ yields
\setlength{\arraycolsep}{0.14em}
\begin{IEEEeqnarray}{lcl}
 \partial_{\theta_i}B_k \cdot B_k^{+} & = &
\partial_{\theta_i}Q_k \left( A_k B_k^{+}\right) + Q_k \cdot \partial_{\theta_i}A_k \cdot
B_k^{+} \nonumber \\
&  = &
\partial_{\theta_i}Q_k \left( Q_k^T  B_k B_k^{+}\right) + \left( Q_k \cdot \partial_{\theta_i}A_k\right)
B_k^{+}. \label{grad:3}
\end{IEEEeqnarray}
\setlength{\arraycolsep}{5pt}

One can easily obtain the explicit expression for $B_k^{+}$:
\begin{equation}
\label{pseudoinverse} B_k^{+}=
\left[
\begin{IEEEeqnarraybox}[][c]{c/c/c}
R_{e,k}^{-1/2} & -R_{e,k}^{-1/2} \bar K_{p,k}^T P_{k+1}^{-1/2} & 0    \\
0 & P_{k+1}^{-1/2} &  0  \\
0 & 0 &   0
\end{IEEEeqnarraybox}
\right].
\end{equation}

By using~(\ref{gr-esrcf:1}), we replace $Q_k \cdot
\partial_{\theta_i}A_k$ in~(\ref{grad:3}) by the quantities already computed. Then,
 taking into account~(\ref{pseudoinverse}), we derive
the equation for the $(m+n)\times (m+n)$ main block of the matrix
$B_k$:
\setlength{\arraycolsep}{0.14em}
\begin{IEEEeqnarray}{r}
 \left[
\begin{IEEEeqnarraybox}[][c]{c/c}
\partial_{\theta_i} R_{e,k}^{1/2} & \partial_{\theta_i} \bar K_{p,k}^T  \\
0 & \partial_{\theta_i} P_{k+1}^{1/2}
\end{IEEEeqnarraybox} \right]
\left[
\begin{IEEEeqnarraybox}[][c]{c/c}
 R_{e,k}^{1/2} &  \bar K_{p,k}^T  \\
0 & P_{k+1}^{1/2}
\end{IEEEeqnarraybox} \right]^{-1}
  =   \left[ \partial_{\theta_i} Q_k \times Q_k^T \right]_{m+n}
\nonumber \\
\label{grad:4} +  \left[
\begin{IEEEeqnarraybox}[][c]{c/c}
X_i & Y_i \\
N_i & V_i
\end{IEEEeqnarraybox}
\right] \left[
\begin{IEEEeqnarraybox}[][c]{c/c}
 R_{e,k}^{1/2} & \bar K_{p,k}^T  \\
0 & P_{k+1}^{1/2}
\end{IEEEeqnarraybox} \right]^{-1}
\end{IEEEeqnarray}
\setlength{\arraycolsep}{5pt}
 where $\left[ \partial_{\theta_i} Q_k
\cdot Q_k^T \right]_{m+n}$ denotes the $(m+n)\times (m+n)$ main
block of the matrix $\partial_{\theta_i} Q_k \cdot Q_k^T$.

As discussed in~\cite{BBVP1990}, the matrix $\partial_{\theta_i} Q_k
\cdot Q_k^T$ is skew symmetric and, hence, can be represented in the
form $\bar L^T - \bar L$ where $\bar L$ is strictly lower
triangular.

Now, let us consider matrix equation~(\ref{grad:4}). As can be seen,
the matrix on the left-hand side of~(\ref{grad:4}) is block upper
triangular. Thus, the strictly lower triangular part of the matrix
$\left[\partial_{\theta_i} Q_k \cdot Q_k^T \right]_{m+n}$ must
exactly cancel the strictly lower triangular part of the second term
on the right-hand side of~(\ref{grad:4}). In other words, if
\begin{displaymath}
 \left[
\begin{IEEEeqnarraybox}[][c]{c/c}
X_i & Y_i \\
N_i & V_i
\end{IEEEeqnarraybox}
\right] \left[
\begin{IEEEeqnarraybox}[][c]{c/c}
 R_{e,k}^{1/2} & \bar K_{p,k}^T  \\
0 & P_{k+1}^{1/2}
\end{IEEEeqnarraybox} \right]^{-1} = \bar L_i + D_i + \bar U_i,
\end{displaymath} then
\begin{equation}
\label{grad:5}
 \left[
\partial_{\theta_i} Q_k \cdot Q_k^T \right]_{m+n} =\bar L_i^T-\bar
L_i.
\end{equation}

Substitution of~(\ref{grad:5}) into~(\ref{grad:4}) leads to the
result
\begin{equation}
\label{grad:7}
 \left[
\begin{IEEEeqnarraybox}[][c]{c/c}
\partial_{\theta_i} R_{e,k}^{1/2} & \partial_{\theta_i} \bar K_{p,k}^T  \\
0 & \partial_{\theta_i} P_{k+1}^{1/2}
\end{IEEEeqnarraybox} \right]
= \left[ \bar L_i^T+ D_i+\bar U_i \right] \left[
\begin{IEEEeqnarraybox}[][c]{c/c}
 R_{e,k}^{1/2} &  \bar K_{p,k}^T  \\
0 &  P_{k+1}^{1/2}
\end{IEEEeqnarraybox} \right].
\end{equation}

Formulas~(\ref{grad:7}) and (\ref{grad:5}) are, in fact,
equations~(\ref{gr-esrcf:3}) and (\ref{gr-esrcf:2}) of the proposed
method for score evaluation. The theorem is half proved.

{\bf Part II.} We need to verify~(\ref{gr-esrcf:4}). By
differentiating the last equation of the eSRCF with respect to the
components of $\theta$
$$
Q_k \left[
\begin{IEEEeqnarraybox}[][c]{c}
-R_k^{-T/2}z_k \\
 P_k^{-T/2} \hat x_k \\
 0
\end{IEEEeqnarraybox}
\right] = \left[
\begin{IEEEeqnarraybox}[][c]{c}
- \bar e_k \\
 P_{k+1}^{-T/2}\hat x_{k+1} \\
\gamma_k
\end{IEEEeqnarraybox}
\right]
$$
we obtain
\setlength{\arraycolsep}{0.14em}
\begin{IEEEeqnarray}{r}
 \left[
\begin{IEEEeqnarraybox}[][c]{c}
-\partial_{\theta_i}{\bar e_k } \\
\partial_{\theta_i}{\left( P_{k+1}^{-T/2}\hat x_{k+1} \right) } \\
\partial_{\theta_i}{\gamma_k }
\end{IEEEeqnarraybox}
\right]  =  \partial_{\theta_i}{Q_k}\cdot Q_k^T \cdot Q_k \left[
\begin{IEEEeqnarraybox}[][c]{c}
-R_k^{-T/2}z_k \\
 P_k^{-T/2} \hat x_k \\
 0
\end{IEEEeqnarraybox}
\right]\phantom{.} \nonumber \\
\label{form:2*}  + Q_k \left[
\begin{IEEEeqnarraybox}[][c]{c}
-\partial_{\theta_i}{\left( R_k^{-T/2}z_k \right) }\\
\phantom{-}\partial_{\theta_i}{\left( P_k^{-T/2} \hat x_k  \right) } \\
 0
\end{IEEEeqnarraybox}
\right].
\end{IEEEeqnarray}
\setlength{\arraycolsep}{5pt}

Next, we replace the last term in~(\ref{form:2*}) with the
quantities already computed and collected in the right-hand side
matrix of~(\ref{gr-esrcf:1}). Furthermore, it is useful to note that
the element $\partial_{\theta_i}{\gamma_k}$ is of no interest here.
These two steps give us
\setlength{\arraycolsep}{0.14em}
\begin{IEEEeqnarray}{lcl}
\!\!\!\left[
\begin{IEEEeqnarraybox}[][c]{c}
-\partial_{\theta_i}{\bar e_k} \\
\partial_{\theta_i}{\left( P_{k+1}^{-T/2}\hat x_{k+1} \right) }
\end{IEEEeqnarraybox}
\right] & = & \Bigl[\partial_{\theta_i}{Q_k}\cdot Q_k^T\Bigl]_{m+n}
\left[
\begin{IEEEeqnarraybox}[][c]{c}
-\bar e_k\\
P_{k+1}^{-T/2}\hat x_{k+1}
\end{IEEEeqnarraybox}
\right] \nonumber \\
\label{form:2} & + & \Bigl[\partial_{\theta_i}{Q_k}\cdot
Q_k^T\Bigl]_{col:~last~q}^{row:~1:m+n} \gamma_k + \left[
\begin{IEEEeqnarraybox}[][c]{c}
M_i\\
W_i
\end{IEEEeqnarraybox}
\right]
\end{IEEEeqnarray}
\setlength{\arraycolsep}{5pt} where
$\Bigl[\partial_{\theta_i}{Q_k}\cdot
Q_k^T\Bigl]_{col:~last~q}^{row:~1:m+n}$ stands for the $(m+n)\times
q$ matrix composed of the entries that are located at the
intersections of the last $q$ columns with the first $m+n$ rows of
$\partial_{\theta_i}{Q_k}\cdot Q_k^T$.

Taking into account~(\ref{grad:5}), from the equation above we
obtain \setlength{\arraycolsep}{0.14em}
\begin{IEEEeqnarray}{lcl}
\!\! \left[
\begin{IEEEeqnarraybox}[][c]{c}
-\partial_{\theta_i}{\bar e_k} \\
\partial_{\theta_i}{\left( P_{k+1}^{-T/2}\hat x_{k+1} \right) }
\end{IEEEeqnarraybox}
\right]  & = & \Bigl[\bar L_i^T - \bar L_i \Bigl] \left[
\begin{IEEEeqnarraybox}[][c]{c}
-\bar e_k\\
P_{k+1}^{-T/2}\hat x_{k+1}
\end{IEEEeqnarraybox}
\right] \nonumber \\
\label{grad:8} & + & \Bigl[\partial_{\theta_i}{Q_k}\cdot
Q_k^T\Bigl]_{col:~last~q}^{row:~1:m+n} \gamma_k + \left[
\begin{IEEEeqnarraybox}[][c]{c}
M_i\\
W_i
\end{IEEEeqnarraybox}
\right]
\end{IEEEeqnarray}
\setlength{\arraycolsep}{5pt}
 where $\bar L_i$ is strictly lower
triangular part of the matrix in~(\ref{gr-esrcf:2}).

Since $\partial_{\theta_i}{Q_k}\cdot Q_k^T$ is skew symmetric, we
can write down
\begin{equation}
\label{new_part}
 \Bigl[\partial_{\theta_i}{Q_k}\cdot
Q_k^T\Bigl]_{col:~last~q}^{row:~1:m+n}=
-\left[\Bigl[\partial_{\theta_i}{Q_k}\cdot
Q_k^T\Bigl]_{col:~1:m+n}^{row:~last~q}\right]^T
\end{equation}
where $\Bigl[\partial_{\theta_i}{Q_k}\cdot
Q_k^T\Bigl]_{col:~1:m+n}^{row:~last~q}$ stands for the
$q\times(m+n)$ matrix composed of the entries that are located at
the intersections of the last $q$ rows with the first $(m+n)$
columns of $\partial_{\theta_i}{Q_k}\cdot Q_k^T$.

To evaluate the right-hand side of~(\ref{new_part}), we return
to~(\ref{grad:3}) and write it in the matrix form:
\setlength{\arraycolsep}{0.14em}
\begin{IEEEeqnarray}{r}
 \left[
\begin{IEEEeqnarraybox}[][c]{c/c/c}
\partial_{\theta_i} R_{e,k}^{1/2} & \partial_{\theta_i} \bar K_{p,k}^T  & 0 \\
0 & \partial_{\theta_i} P_{k+1}^{1/2} & 0 \\
0 & 0 & 0
\end{IEEEeqnarraybox} \right] \left[
\begin{IEEEeqnarraybox}[][c]{c/c/c}
 R_{e,k}^{1/2} \; & \; \bar K_{p,k}^T \; & \;  0  \\
0 \; & \;  P_{k+1}^{1/2}\;  & \; 0 \\
0 \; & \; 0 \; & \; 0
\end{IEEEeqnarraybox} \right]^{+}
  =    \partial_{\theta_i} Q_k \cdot Q_k^T \nonumber \\
 \label{matrix_form} +  \left[
\begin{IEEEeqnarraybox}[][c]{c/c/c}
X_i \; &\; Y_i\; &\; 0 \\
N_i \; &\; V_i\; &\; 0\\
B_i \; &\; K_i\; &\; 0
\end{IEEEeqnarraybox}
\right] \left[
\begin{IEEEeqnarraybox}[][c]{c/c/c}
 R_{e,k}^{1/2} \; &  \; \bar K_{p,k}^T \;  &\;  0\\
0 \; & \; P_{k+1}^{1/2} \;  &\; 0 \\
0 \; & \;  0 \;  &\; 0
\end{IEEEeqnarraybox} \right]^{+}.
\end{IEEEeqnarray}
\setlength{\arraycolsep}{5pt}

As can be seen, the last (block) row of the left-hand side matrix
in~(\ref{matrix_form}) is zero. Thus, the last (block) row of the
matrix $\partial_{\theta_i} Q_k \cdot Q_k^T$ must exactly cancel the
last (block) row of the second term in~(\ref{matrix_form}):
\begin{equation}
\label{new_part:1} \Bigl[\partial_{\theta_i}{Q_k}\cdot
Q_k^T\Bigl]_{col:~1:m+n}^{row:~last~q} = -\left[
\begin{IEEEeqnarraybox}[][c]{c/c} B_i & K_i
\end{IEEEeqnarraybox}
\right] \left[
\begin{IEEEeqnarraybox}[][c]{c/c}
 R_{e,k}^{1/2} & \bar K_{p,k}^T \\
0 & P_{k+1}^{1/2}
\end{IEEEeqnarraybox} \right]^{-1}.
\end{equation}

By substituting~(\ref{new_part:1}) into~(\ref{new_part}), we obtain
\begin{equation}
\label{form:1} \Bigl[\partial_{\theta_i}{Q_k}\cdot
Q_k^T\Bigl]_{col:~last~q}^{row:~1:m+n} =
  \left[
\begin{IEEEeqnarraybox}[][c]{cc}
R_{e,k}^{1/2} & \bar K_{p,k}^T  \\
0 & P_{k+1}^{1/2}
\end{IEEEeqnarraybox}
\right]^{-T} \left[
\begin{IEEEeqnarraybox}[][c]{c}
B_i \\ K_i
\end{IEEEeqnarraybox}
\right].
\end{equation}

Final substitution of~(\ref{form:1}) into ~(\ref{grad:8})
validates~(\ref{gr-esrcf:4}) of the proposed method for the Log LG
evaluation. This completes the proof.
\end{proof}

\begin{remark}
\label{rem:1} The method for score evaluation introduced above has
been derived from the eSRCF implementation. As a consequence, the
proposed method is of covariance-type.
\end{remark}

\begin{remark}
The new square-root algorithm for score evaluation naturally extends
the eSRCF filter and, hence, consists of two parts. They are the
"filtered"  and "differentiated" parts. This structure allows the
Log LF and its gradient to be computed simultaneously. Thus, the
method is ideal for simultaneous state estimation and parameter
identification.
\end{remark}

\begin{remark}
In the KF formulation of the Log LG evaluation, it is necessary to
run the "differentiated" KF for each of the parameters $\theta_i$ to
be estimated. As in~\cite{BBVP1990}, in the eSRCF formulation this
"bank" of filters is replaced with the augmented arrays to which
orthogonal transformations are applied.
\end{remark}

\begin{table*}
\renewcommand{\arraystretch}{1.3}
\caption{Effect of roundoff errors on the computed solutions for the
set of test problems from Example~\ref{ex:3}} \label{tab:1}
\centering
\begin{tabular}{c|c||c|c|c|c||c|c|c|c}
\hline \multicolumn{2}{c||}{Problem conditioning} &
\multicolumn{4}{|c||}{Conventional KF technique} &
\multicolumn{4}{|c}{Suggested square-root method} \\
\hline $\delta$ & $K(R_{e,1})$ & $\Delta  P_1$ &  $\Delta P'_1$ &
$\Delta Log LF$ & $\Delta Log LG$ & $\Delta  P_1$ &  $\Delta P'_1$ &
$\Delta Log LF$ & $\Delta Log LG$ \\
\hline
 $ 10^{-2\phantom{0}}$ & $10^{3\phantom{0}}$ & $1 \cdot 10^{-13}$           & $1 \cdot 10^{-10}$            & $2 \cdot 10^{-13}$              & $1\cdot 10^{-13}$          & $4\cdot 10^{-15}$           & $7\cdot 10^{-16}$               &       $1\cdot 10^{-13}$   &$9\cdot 10^{-14}$ \\
 $ 10^{-4\phantom{0}}$ & $10^{7\phantom{0}}$ & $5 \cdot 10^{-10}$           & $9 \cdot 10^{-4\phantom{0}}$  & $4 \cdot 10^{-9\phantom{0}}$    & $1\cdot 10^{-9\phantom{0}}$& $4\cdot 10^{-13}$           & $7\cdot 10^{-14}$               & $6\cdot 10^{-10}$ & $7\cdot 10^{-10}$ \\
 $ 10^{-6\phantom{0}}$ & $10^{11}$           & $2 \cdot 10^{-6\phantom{0}}$ & $2 \cdot 10^{-1\phantom{0}}$   &  $2 \cdot 10^{-5\phantom{0}}$  & $6\cdot 10^{-6\phantom{0}}$& $3\cdot 10^{-11}$           & $1\cdot 10^{-11}$               & $9\cdot 10^{-6\phantom{0}}$ &  $4\cdot 10^{-6\phantom{0}}$ \\
 $ 10^{-8\phantom{0}}$ & $10^{15}$           & $3 \cdot 10^{-3\phantom{0}}$ & $2 \cdot 10^{-1\phantom{0}}$   &  $3 \cdot 10^{-1\phantom{0}}$  & $2\cdot 10^{-2\phantom{0}}$& $3\cdot 10^{-10}$           & $2\cdot 10^{-10}$               & $2\cdot 10^{-1\phantom{0}}$ & $9\cdot 10^{-3\phantom{0}}$ \\
 $ 10^{-9\phantom{0}}$ & $10^{16}$           & $3 \cdot 10^{-1\phantom{0}}$ & NaN                            &  $4 \cdot 10^{0\phantom{-0}}$  & $4\cdot 10^{0\phantom{-0}}$ & $2\cdot 10^{-8\phantom{0}}$ &  $7\cdot 10^{-9\phantom{0}}$    & $1\cdot 10^{0\phantom{-0}}$ & $5\cdot 10^{1\phantom{-0}}$ \\
 $ 10^{-10}$           & $\infty$            &   NaN                        &       NaN                      & NaN                            & NaN                        & $2\cdot 10^{-7\phantom{0}}$ & $1\cdot 10^{-8\phantom{0}}$     &  $2\cdot 10^{4\phantom{-0}}$    & $2\cdot 10^{4\phantom{-0}}$ \\
\hline
\end{tabular}
\end{table*}

\section{Numerical Results \label{s6}}

First, we would like to check our theoretical derivations. To do so,
we apply the square-root algorithm introduced in
Theorem~\ref{theorem:1} to the following simple test problem.

\begin{exmp}
\label{ex:1} Consider the special case of the system~(1), (2) being
\begin{displaymath}
x_{k}= \left [
\begin{IEEEeqnarraybox}[][c]{c}
d_{k}  \\
s_{k}
\end{IEEEeqnarraybox} \right ]= \left [
\begin{IEEEeqnarraybox}[][c]{c/c/c}
1 & \Delta t  \\
0 & \  \  e^{-\Delta t  / \tau}
\end{IEEEeqnarraybox}
 \right ] x_{k-1}+ \left [
\begin{IEEEeqnarraybox}[][c]{c}
0  \\
1
\end{IEEEeqnarraybox}
 \right ] w_k, \;
 z_k= \left [
\begin{IEEEeqnarraybox}[][c]{c/c}
1 &\ 0
\end{IEEEeqnarraybox} \right ] x_k+v_k
\end{displaymath}
 where $w_k \sim N(0,I_2)$, $v_k \sim N(0,I_1)$, $I_n$ denotes the
$n\times n$ identity matrix and $\tau$ is a parameter which needs to
be estimated.
\end{exmp}

In our simulation experiment, we compute the negative Log LF and its
gradient by the proposed square-root method and, then, compare the
results to those produced by the conventional KF approach. The
outcomes of this experiments are illustrated by Fig.~\ref{fig1_a}
and \ref{fig1_b}.

As can be seen from Fig.~\ref{fig1_b}, all algorithms for score
evaluation produce exactly the same result and give the same zero
point that further coincides with the minimum point of the negative
Log LF (see Fig.~\ref{fig1_a}). All these evidences substantiate the
theoretical derivations of Section~\ref{score-complex}.

\begin{figure}
\centering
\includegraphics[width=2.5in]{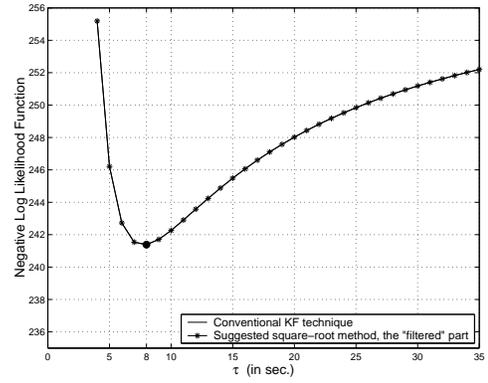}
\caption{The negative Log LF computed by the eSRCF and the
conventional KF for Example~\ref{ex:1}} \label{fig1_a}
\end{figure}

\begin{figure}
\centering
\includegraphics[width=2.5in]{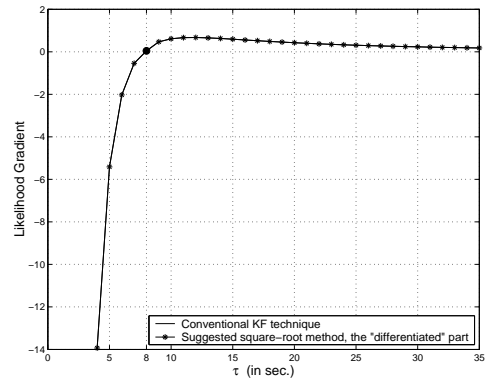}
\caption{ The Log LG computed by the proposed square-root method and
the conventional KF for Example~\ref{ex:1}} \label{fig1_b}
\end{figure}

Next, we wish to answer the second question posed in this paper:
does the algorithm for score evaluation derived from  numerically
stable square-root implementation method improve the robustness of
computations against roundoff errors? The previously obtained
results (Example~\ref{ex:1}) indicate that both methods, i.e. the
conventional KF technique and the new square-root algorithm, produce
exactly the same answer for the Log LF and Log LG evaluation.
However, numerically they no longer agree. We are now going to
explore the accuracy of the numerical algorithms.

To begin designing the ill-conditioned test problem we, first,
stress the type of the proposed method. As discussed in
Remark~\ref{rem:1}, the new square-root algorithm belongs to the
class of covariance-type methods. From Verhaegen and Van Dooren's
celebrated paper~\cite{VerhaegenDooren1986}, we know that the
condition number of the innovation covariance matrix $K(R_{e,k})$ is
the key parameter determining the numerical behavior of the
covariance algorithms. Taking into account these two important
facts, we construct the following ill-conditioned test problem.

\begin{exmp}
\label{ex:3} Consider the problem with the measurement sensitivity
matrix
\begin{displaymath} H_k=\left[
      \begin{IEEEeqnarraybox}[][c]{c/c/c}
      1 & 1 & 1 \\
      1 & 1 & 1+\delta
      \end{IEEEeqnarraybox}
     \right] \mbox{ and }
F_k=I_3, G_k=0, Q_k=I_1, R_k=\delta^2\theta I_2
\end{displaymath}
with $x_0 \sim {\cal N}(0,\theta I_3)$, where  $\theta$ is an
unknown system parameter. To simulate roundoff we assume that
$\delta^2<\epsilon_{roundoff}$, but $\delta>\epsilon_{roundoff}$
where $\epsilon_{roundoff}$ denotes the unit roundoff
error\footnote{Computer roundoff for floating-point arithmetic is
often characterized by a single parameter $\epsilon_{roundoff}$,
defined in different sources as the largest number such that either
$1+\epsilon_{roundoff} = 1$ or $1+\epsilon_{roundoff}/2 = 1$ in
machine precision. }.
\end{exmp}

When $\theta=1$, Example~\ref{ex:3} coincides with well-known
ill-conditioned filtering problem (see, for
instance,~\cite{Grewal2001}) and demonstrates how a problem that is
well-conditioned, as posed, can be made ill-conditioned by the
filter implementation.
 The difficulty to be explored is in matrix
inversion. As can be seen, although $rank \  H=2$, the matrix
$R_{e,1}$ is singular in machine precision that yields the failure
of the  conventional KF implementation. We introduced an unknown
system parameter $\theta$ making sure that the same problem is
applied to the matrix $(R_{e,1})'_{\theta}$ for each value of
$\theta$. Thus, both parts of the method for score evaluation, that
are the "filtered" and "differentiated" parts, fail after processing
the first measurement. From the discussion above we understand that
Example~\ref{ex:3} demonstrates the difficulty only for the
covariance-type methods.

Our simulation experiments presented below are organized as follows.
All methods were implemented in the same precision (64-bit floating
point) in MatLab where the unit roundoff error is $2^{-53}\approx
1.11 \cdot 10^{-16}$. The MatLab function $eps$ is twice the unit
roundoff error and $\delta=eps^{2/3}$ satisfies the conditions
$\delta^2<\epsilon_{roundoff}$ and $\delta>\epsilon_{roundoff}$ from
Example~\ref{ex:3}. We provide the computations for for different
values of $\delta$, say $\delta \in [10^{-9}eps^{2/3},
10^{9}eps^{2/3}]$. This means that we consider a set of test
problems from Example~\ref{ex:3}. The unknown system parameter
$\theta$ is fixed, say $\theta=2$. The exact answers are produced by
the Symbolic Math Toolbox of MatLab.

{\it Experiment 1:} In this experiment we are going to use the {\it
performance profile technique} to compare the conventional KF
approach for score evaluation with the square-root algorithm
introduced in this paper. The performance profile method was
developed by Dolan and Mor$\rm \acute{e}$~\cite{Dolan} to answer a
common question in scientific computing: how to compare several
competing methods on a set of test problem. Now, it can be found in
textbooks (see, for instance,~\cite{higham}).

In our simulation experiments we consider a set $A$ of $n=2$
algorithms, mentioned above. The performance measure, $t_a(p)$, is a
measure of accuracy. More precisely, $t_a(p)$ is the maximum
absolute error in Log LG computed for $7$ different values of
$\delta$. Thus, we consider a set $P$ of $m=7$ test problems from
Example~\ref{ex:3}; $\delta \in [10^{-2}, 10^{-3}, 10^{-4}, 10^{-5},
10^{-6}, 10^{-7}, 10^{-8}]$. According to the performance profile
technique, we compute the performance ratio
\begin{displaymath}
r_{p,a} = \frac{t_a(p)}{\min\{ t_{\sigma}(p): \sigma \in A\}} \ge 1,
\end{displaymath}
which is the performance of algorithm  $a$ on problem $p$ divided by
the best performance of all the methods (we mean a particular
implementation method for score evaluation) on this problem. The
performance profile of algorithm $a$ is the function
\begin{displaymath}
 \phi_a(\mu) =
\frac{1}{m} \times \mbox{ number  of } p \in P \mbox{ such  that  }
r_{p,a} \le \mu,
\end{displaymath} which is monotonically
increasing. Thus, $\phi_a(\mu)$ is the probability that the
performance of algorithm $a$ is within a factor $\mu$ of the best
performance over all implementations on the given set of problems.

\begin{figure}
\centering
\includegraphics[width=2.5in]{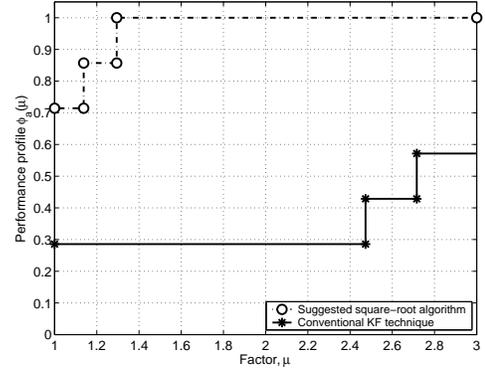}
\caption{Performance profiles of the methods for score evaluation:
the conventional KF implementation and the new square-root algorithm
proposed in this paper, -- on the set of test problems from for
Example~\ref{ex:3}. } \label{f:2}
\end{figure}

The results of this experiment are illustrated by Fig.~\ref{f:2}.
For each method, $\mu$ is plotted against the performance profile
$\phi_a(\mu)$, for $\mu \in [0, 3]$. We are now going to explain
Fig.~\ref{f:2}.

 Let us consider the left-hand side of
Fig.~\ref{f:2}, where $\mu=1$. We can say that the new square-root
algorithm proposed in this paper is the most accurate implementation
on $\approx 71\%$ of the problems, with the conventional KF being
accurate on $30\%$ of the problems. Next, we consider the middle of
the plot, looking where the curve first hit probability 1. We
conclude that the suggested square-root method is within a factor
$\mu \approx 1.3$ of being the most accurate implementation on every
test problem. However, the conventional KF approach for score
evaluation will never manage all $7$ problems (as $\delta \to
\epsilon_{roundoff}$, the machine precision limit, the test problems
become ill-conditioned). We need to increase $\mu$ to $\approx 2.7$
to be able to say that for $\approx 58\%$ of the test problems the
conventional KF provides an accurate Log LG evaluation within a
factor $\mu \approx 2.7$.

Thus, the performance profiles clearly indicate that on the set of
the test problems from Example~\ref{ex:3} the new square-root
algorithm derived in this paper provides more accurate evaluation of
the Log LG compared with the conventional KF approach.

{\it Experiment 2:} In this experiment we use the conventional KF
technique and the proposed square-root method to compute the maximum
absolute error in Log LF, denoted as $\Delta Log LF$, and its
gradient, denoted as $\Delta Log LG$. The results of this experiment
 are summarized in Table~\ref{tab:1}. We also present the maximum
absolute error among elements in matrices $P_1$ and
$\left(P_1\right)'_{\theta}$ (denoted as $\Delta P_1$ and $\Delta
P^{\prime}_1$, respectively) to explore the numerical behavior of
the "filtered" and "differentiated" parts of the methods for score
evaluation.

As can be seen from Table~\ref{tab:1}, the square-root
implementation of the Riccati-type sensitivity equation  degrades
more slowly than the conventional Riccati-type sensitivity recursion
as $\delta \to \epsilon_{roundoff}$, the machine precision limit
(see columns denoted as $\Delta P^{\prime}_1$). For instance, the
"filtered" (columns $\Delta P_1$) and "differentiated" (columns
$\Delta P^{\prime}_1$) parts of the proposed square-root method for
score evaluation maintain about $7$ and $8$ digits of accuracy,
respectively, at $\delta=10^{-9}$. The conventional KF technique
provides essentially no correct digits in both computed solutions.
Besides, it seems that the roundoff errors tend to accumulate and
degrade the accuracies of the Log LF and Log LG faster than the
accuracies of $\Delta P_1$ and $\Delta P^{\prime}_1$. Indeed, for
the same $\delta=10^{-9}$ we obtain no correct digits in the
computed solutions for all methods. In MatLab, the term 'NaN' stands
for 'Not a Number' that actually means the failure of the numerical
algorithm.

\begin{remark}
\label{rem:4} The results of Experiment~2 indicate that the new
square-root algorithm provides more accurate computation of the
sensitivity matrix $\left(P_k\right)'_{\theta}$ compared to the
conventional KF. Hence, it can be successfully used in all
applications where this quantity is required.
\end{remark}

\section{Concluding Remarks\label{conclusion}}

In this paper, a numerically stable square-root implementation
method for KF formulas, the eSRCF, has been extended in order to
compute the Log LG for linear discrete-time stochastic systems. The
preliminary analysis indicates that the new algorithm for score
evaluation provides more accurate computations compared with the
conventional KF approach. The new result can be used for efficient
calculations in sensitivity analysis and in gradient-search
optimization algorithms for the maximum likelihood estimation of
unknown system parameters.

As an extension of the eSRCF, the new method for score evaluation is
expected to inherit its benefits. However, the question about
suitability for parallel implementation is still open.

It can be mentioned that another approach to construct numerically
stable implementation method for score evaluation is to use the $UD$
filter~\cite{Bierman77}. Being the modification of the square-root
implementations, the $UD$-type algorithms improve the robustness of
computations against roundoff errors, but compared with SRF,  the
$UD$ filter reduces the computational cost (see~\cite{Bierman77},
\cite{VerhaegenDooren1986},  \cite{KailathSayed2000}). As mentioned
in~\cite{BBVP1990} and as far as this author knows, it is still not
known how to use the $UD$ filter to compute the score.



\end{document}